\documentclass[conference]{IEEEtran}
\usepackage{caption}
\usepackage{algpseudocode}
\usepackage{amsmath}
\usepackage{amsfonts}
\usepackage{amssymb}
\usepackage{subfigure}
\usepackage{graphicx}
\usepackage[table]{xcolor}
\usepackage{array,multirow,graphics}

\usepackage{amsfonts}
\usepackage{booktabs}
\usepackage{siunitx}

\usepackage{color}

\usepackage{url}
\makeatletter
\def\url@leostyle{%
  \@ifundefined{selectfont}{\def\UrlFont{\sf}}{\def\UrlFont{\small\ttfamily}}}
\makeatother
\newcommand{\eat}[1]{}
\usepackage[linesnumbered,ruled]{algorithm2e}
\usepackage{mdframed} 
\usepackage{enumitem}

\newenvironment{packed_enum}{%
  \begin{enumerate}%
  }{\end{enumerate}}

\usepackage{amsthm}

\newtheorem{lemma}{Lemma}

\newcolumntype{L}[1]{>{\raggedright\let\newline\\\arraybackslash\hspace{0pt}}m{#1}}

\begin{document}

\title{Decentralized Release of Self-emerging Data using Smart Contracts}
\vspace{-3mm}
\author{%
{Chao Li and Balaji Palanisamy }%
\vspace{1.6mm}\\
\fontsize{10}{10}\selectfont\itshape
School of Computing and Information, University of Pittsburgh, USA\\
\{chl205, bpalan\}@pitt.edu\\
\vspace{1.2mm}\\
\fontsize{10}{10}\selectfont\rmfamily\itshape
\vspace{-0.55in}
}%







\maketitle

\begin{abstract}

In the age of Big Data, releasing protected sensitive data at a future point in time is critical for various applications. Such self-emerging data release requires the data to be protected until a prescribed data release time and be automatically released to the recipient at the release time, even if the data sender goes offline. 
While straight-forward centralized approaches provide a basic solution to the problem, unfortunately they are limited to a single point of trust and involve a single point of control. 
This paper presents decentralized techniques for supporting self-emerging data using smart contracts in Ethereum blockchain networks. 
We design a credible and enforceable smart contract for supporting self-emerging data release. 
The smart contract employs a set of Ethereum peers to jointly follow the proposed timed-release service protocol allowing the participating peers to earn the remuneration paid by the service users.
We model the problem as an extensive-form game with imperfect information to protect against possible \textit{post-facto attacks} including some peers destroying the private data (\textit{drop attack}) or secretly releasing the private data before the release time (\textit{release-ahead attack}). 
We demonstrate the efficacy and attack-resilience of the proposed techniques through rigorous analysis and experimental evaluation. 
Our implementation and experimental evaluation on the Ethereum official test network demonstrate the low monetary cost and the low time overhead associated with the proposed approach and validate its guaranteed security properties.   

\end{abstract}

\section{Introduction} 
\label{intro}

In the age of Big Data, releasing protected sensitive data at a future point in time is critical for various applications.
Such self-emerging data release requires the data to be protected until a prescribed data release time and be automatically released to the recipient at the release time, 
even if the data sender goes offline~\cite{kasamatsu2012time,kikuchi2011strong,may1992timed,rivest1996time}.
For example, Alice, who is working in the USA, needs to make some time-sensitive documents (e.g., press releases, online exam papers) be received by Bob exactly at 9 a.m. in Bob's time zone. Considering the time difference, without having to stay online at 3 a.m., Alice may prefer to send out the documents at her most convenient time while still making the documents be released to Bob at 9 a.m.
In another example, Carol, as she is getting old, may intend to draw up a will that can only be released to her families after three years.

Centralized systems such as cloud storage services \cite{cloud1,cloud2,cloud3} may provide a simple and straight-forward approach for implementing self-emerging data release. The service provider may simply keep the sensitive data until the prescribed release time and make it available at the release time. However, such a centralized approach limits the data protection to a single point of trust and a single point of control.
Even in cases when the service providers are trustworthy, such centralized models lead to channels of attacks beyond the control of service providers for an adversary to breach the security and privacy of the data. It includes insider attacks \cite{ins1, ins2}, external attacks on the centralized data infrastructures, malware and large-scale denial-of-service attacks \cite{ext1, ext2}. In 2014, 28\% of the respondents of the US State of Cybercrime Survey \cite{ins1} reported being victims of insider attacks and 32\% reported that insider attacks were more damaging than outsider attacks.

\begin{figure*}[h]
\vspace{-5mm}
\centering
{
    \includegraphics[width=12cm,height=4.5cm]{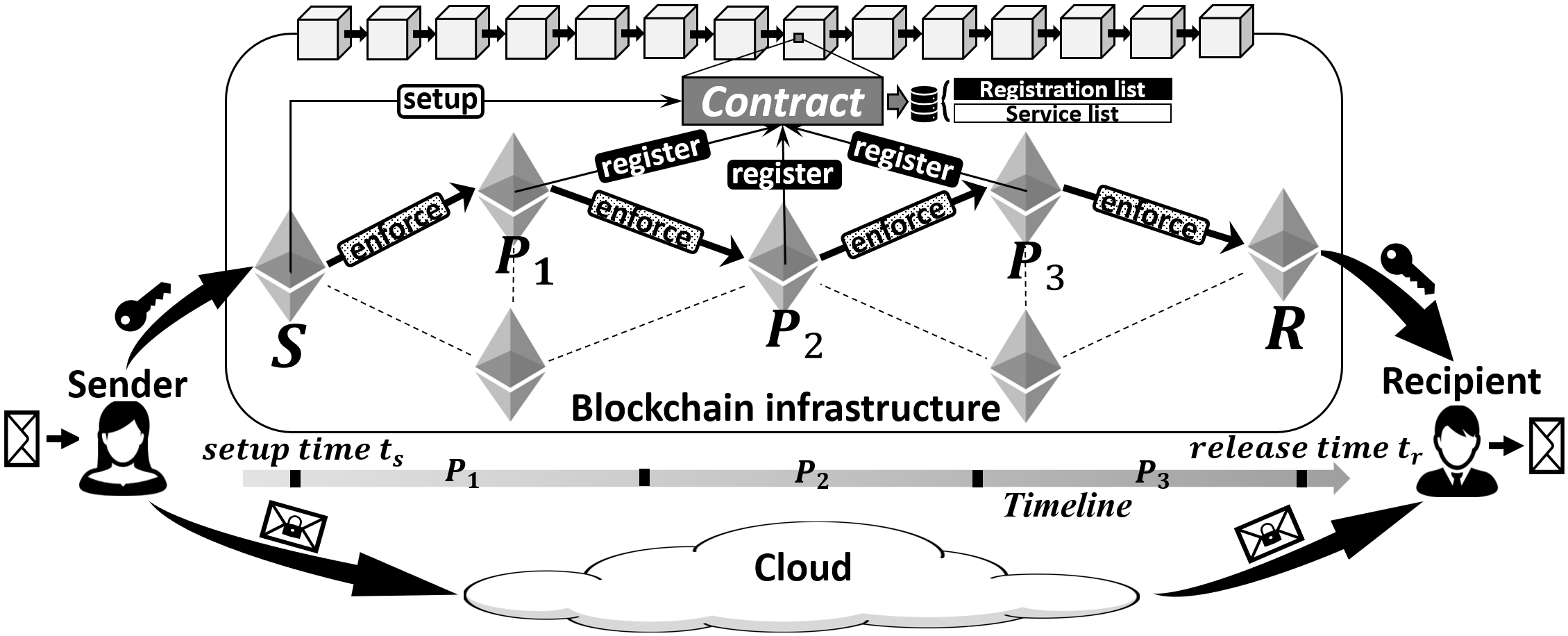}
}
\vspace{-1.5mm}
\caption {Self-emerging data release system}
\vspace{-5mm}
\label{f1} 
\end{figure*}

In this paper, we develop a decentralized self-emerging data release system over Ethereum blockchain networks~\cite{wood2014ethereum} that does not involve a single point of trust or control.
The proposed mechanisms route the self-emerging data within the blockchain infrastructure and enable it to automatically appear at the release time while making it harder for an adversary to access it prior to the release time. The choice of Blockchains as the underlying data infrastructure network is motivated by the facts that Blockchains are huge-scale massively distributed systems that make complete decentralization possible and they are inherently designed to be reliable and robust to failures.
Our smart contract implementation recruits a set of Ethereum peers to jointly follow the proposed timed-release service protocol allowing the participating peers to earn the remuneration paid by the service users. 
Meanwhile, the recruited peers need to pay security deposits so that any detected misbehaviors can result in the deposits being confiscated.
Specifically, we model the problem as an extensive-form game with imperfect information to protect against possible post-facto misbehaviors including some peers destroying the private data (\textit{drop attack}) or secretly releasing the private data before the release time (\textit{release-ahead attack}). Through a careful design of the smart contract based on game theory, we demonstrate that the best choice of any rational Ethereum peer in the proposed technique is to always honestly follow the correct protocol. 
We validate the efficacy and attack-resilience of the proposed techniques through rigorous analysis and experimental evaluation on the Ethereum official test network. The experiments demonstrate the low monetary cost and the low time overhead associated with the proposed approach and validate its guaranteed security properties.

In the rest of the paper, we first introduce the self-emerging data release system in Section \ref{s2}. Then, in Section \ref{s3}, we present the timed-release service protocol in detail. In Section \ref{s4}, we implement and evaluate the proposed protocol on the Ethereum official test network. Finally, we present the related work in Section \ref{s5} and conclude in Section \ref{s6}.

\vspace{-1mm}
\section{System overview}
\vspace{-0.5mm}
\label{s2}
In this section, we present an overview of the proposed self-emerging data release system and we introduce the key ideas behind the proposed timed-release service protocol. 

\vspace{-1mm}
\subsection{Self-emerging data release system}
\label{s2-1}
The proposed self-emerging data release system consists of four key components (Figure~\ref{f1}) namely data senders, data receivers, a cloud storage platform and the blockchain infrastructure enabling the timed data release service.

\noindent \textbf{Data sender (S)}: 
Data senders have private data to be released to data recipients at a future point in time. At \textit{setup time} $t_s$, a data sender encrypts the private data using a secret key, sends the encrypted data to a cloud storage system and sends the encrypted secret key into the blockchain infrastructure for timed release at the expected \textit{release time $t_r$}. 

\noindent \textbf{Data recipient (R)}: 
Data recipients receive the private data at the expected data \textit{release time $t_r$}. While the encrypted private data can be downloaded from the cloud at any time, the secret key from the blockchain infrastructure can be released to data recipients only at $t_r$ determined by data senders.

\noindent \textbf{Cloud}: 
A cloud storage platform is used as a medium for data senders to transfer the encrypted private data to data recipients. 

\noindent \textbf{Blockchain infrastructure}: 
The blockchain infrastructure forms the core component of the self-emerging data release system. It implements the protocols necessary for offering timed-release services to data senders. 

\vspace{-1mm}
\subsection{Timed-release service protocol}
\label{s2-2}

The proposed timed-release service protocol recruits peers from the blockchain peer-to-peer network to store the data during $[t_s,t_r]$ and release the data to the recipients at $t_r$.
The protocol allows any peer to join the system at any time and declare any time period during which they are willing to provide services. 
In case that no single peer can handle the entire $[t_s,t_r]$ time period, the protocol can split $[t_s,t_r]$ into a series of successive shorter time durations, each of which is handled by a different peer.
In the example shown in Figure~\ref{f1}, the storage time duration $[t_s,t_r]$ is split into three fractions and the encrypted secret key is passed from sender $S$ to recipient $R$ through a routing path formed by $P_1$, $P_2$ and $P_3$.
The proposed protocol enables such a routing scheme through onion routing~\cite{dingledine2004tor} that requires the sender to first encrypt the secret key using the public key of the recipient and then iteratively form layers of encryption using the public keys of the selected peers on the routing path. As a result, each peer on the routing path decrypts one layer of the encryption of the secret key using their private keys before forwarding it to the subsequent peers on the path until it reaches the recipient who decrypts the final layer of the encryption to obtain the key in plain text.
The protocol incentivizes the participating peers by requiring the data senders to pay remunerations to the peers for obtaining the store and forward services from them to route the encrypted key along. Also, the protocol requires the participating peers to pay security deposits so that any detected misbehavior can result in their deposits being confiscated. 

The timed-release service protocol satisfies two key requirements in order to be effective in practice.
First, it ensures \textit{credibility} so that senders, recipients and peers are guaranteed that they all see the same protocol when they participate in the timed-release service. We implement the timed-release service protocol using the Ethereum smart contract platform ~\cite{wood2014ethereum} which ensures that when smart contracts get deployed into the blockchain infrastructure, the protocol can be recorded in the blockchain and be available to the public and becomes nearly tamper-proof unless someone controls a majority of computation power of the distributed network~\cite{Ethernodes2017}. 
Second, the protocol needs to be \textit{enforceable} so that peers are guaranteed to receive remunerations for honestly performing the agreed services while being penalized for any misbehavior or failure to render the promised service. In our approach, the protocol forces the participants to pass the ownership of their money to the smart contract such that it ensures that the only way to receive payment from the smart contract is to trigger the contract with a satisfied condition dictated in the protocol.


The proposed protocol consists of four key components which are briefly introduced here and we will present their detailed design in Section \ref{s3}.

\noindent \textbf{Peer registration}:
At any point in time, a new peer \footnote{A peer in this work refers to an externally owned account in Ethereum.} $P$  can register by paying a security deposit to the contract to be added into the \textit{registration list} maintained by the contract.
This process makes the entire network learn that the peer has registered and can provide services during its prescribed working times.
For example, in Figure~\ref{f1}, we find that  $P_1$, $P_2$ and $P_3$ have been registered before the setup time $t_s$.

\vspace{1mm}
\noindent \textbf{Service setup}:
At any point in time, a sender $S$ can pay remunerations and submit peers selected from the registration list to a contract $C$ and set up a timed-release service. This process makes the service to be recorded by a \textit{service list} maintained by the contract. 
In Figure~\ref{f1}, we find that sender $S$ requests a service at $t_s$ with selected peers $P_1$, $P_2$ and $P_3$.

\vspace{1mm}
\noindent \textbf{Service enforcement}:
After a service has been set up, the participants, namely sender $S$, recipient $R$ and peers $P$s should follow the protocol honestly in order to render the service successfully. Behaviors violating the protocol will lead to service failure and such misbehaviors are detected and penalized by the contract.
In Figure~\ref{f1}, the process of routing the encrypted secret key from $S$ to $R$ through the path formed by the three peers is enforced by the contract $C$ through paying remunerations for honest behaviors while confiscating deposits for misbehaviors detected by $C$. 

\vspace{1mm}
\noindent \textbf{Reporting mechanism}:
To effectively detect misbehaviors in the protocol implemented in the smart contract, the reporting mechanism incentivizes peers to report misbehaviors by announcing an award in the contract.

\subsection{Attack models}
\label{s2-3}
In our work, we model adversaries with rationality and consider two key \textit{post-facto} attack models, namely \textit{drop attack} and \textit{release-ahead attack}.

\vspace{1mm}
\noindent \textbf{Rational adversaries}:
Recently, it has been widely recognized that assuming an adversary to be semi-honest or malicious is either too weak or too strong in many practical cases and hence modeling adversaries with rationality~\cite{dong2017betrayal,groce2012fair,guo2016rational,nguyen2013analyzing} is a relevant choice in several attack scenarios. Informally, a semi-honest adversary follows the prescribed protocol but tries to glean more information from available intermediate results while a malicious adversary can take any action for launching attacks~\cite{hazay2010note,zhang2013cryptanalysis}. A rational adversary lies in the middle of the two types. 
That is, rational adversaries are self-interest-driven, they choose to violate protocols, such as colluding with other parties, only when doing so brings them a higher profit. 
In this paper, in order to design our system with strong and practical security guarantees, we model all involved participants, namely $S$, $R$ and $P$, to be rational adversaries without assuming any of them to be honest. 

\vspace{1mm}
\noindent \textbf{Post-facto attacks}:
The system targets post-facto attacks, namely the attacks launched after the data senders decide to release their private data.
In many use cases, data senders, as the source of the private data and the initiator of the process, can determine whether to release the data and when the data should be released. For example, at a certain time point, Carol may decide to draw up a will before anyone else knows her plan. 
Then, by treating this time point as the registration deadline $t_d$ and only selecting peers from the registration list that were registered before $t_d$, it can be guaranteed that all the selected peers were not intentionally registered for attacking just her data.
In the rest of the paper, we assume that such a registration deadline $t_d$ exists, which allows us to focus on the more common and severe attack models, namely drop attack and release-ahead attack through peer bribery. 


\vspace{1mm}
\noindent \textbf{Drop attack}:
A drop attack happens when the encrypted secret key fails to reach the recipient $R$ at release time $t_r$. 
For example, in Figure~\ref{f1}, after receiving the encrypted secret key from peer $P_2$, peer $P_3$ may decide to destroy it.  
In post-facto attacks, due to the existence of the security deposit, a rational peer has no motivation to destroy the data. However, we notice that a drop attack can happen when an adversary intends to bribe the selected peer (say, $P_3$). Specifically, a drop attack can be successful when the rational adversary gets higher profit from the drop attack than the paid bribery and when the bribee receives higher bribery than the drop penalty.
To break the win-win situation, we carefully design the detection mechanism in Section~\ref{s3-3} to make drop attacks detectable and to allow the reporting mechanism in Section~\ref{s3-4} to distinguish and penalize the adversaries. In addition, by modeling the protocol as an extensive-form game with imperfect information~\cite{leyton2008essentials}, we demonstrate that drop attack can be entirely prevented in our rational model.

\vspace{1mm}
\noindent \textbf{Release-ahead attack}:
In release-ahead attacks, an adversary aims to obtain the secret key before the actual release time $t_r$ and earn a profit by utilizing the data prior to the release time. 
In Figure~\ref{f1}, peer $P_3$ can launch a release-ahead attack by releasing the encrypted secret key to recipient $R$ before $t_r$.
Similar to drop attacks, release-ahead attacks may happen through peer bribery in post-facto attacks. 
However, unlike drop attacks that can be detected, a release-ahead attack happens secretly as peers on the path can share stored data to any party without leaving a mark. 
Our proposed techniques handle this challenge by designing a reporting mechanism to model the release-ahead attack as an extensive-form game with imperfect information (Section~\ref{s3-4}). It makes rational adversaries choose to never launch release-ahead attacks as the game ensures that the best choice of any rational Ethereum peer is to always honestly follow the correct protocol. 

\vspace{-1mm}
\subsection{Assumptions}
\vspace{-1mm}
We make the following key assumptions in this paper:
\begin{itemize}[leftmargin=*]
\item We assume that the monetary value of the private data is known to the sender $S$. That is, the maximum profit made by an adversary from the two attack models, without considering the deposit penalty, is bounded by this value. 
\item 
Although techniques such as mixing~\cite{bonneau2014mixcoin} have been proposed, it is still unclear whether identification of peers can be adequately protected in the Ethereum network. Therefore, we assume that adversaries have the ability to communicate with any Ethereum peers and we assume no protection of pseudonymity or anonymity.
\item We also assume that an adversary and the peers in communication do not trust each other as otherwise their cost-free collusion violates the rationality assumption of all parties.
\item Our system employs the use of Whisper protocol~\cite{Whisper2017} to enable communication between two Ethereum peers. We assume that a private channel generated using the Whisper protocol between any two Ethereum peers is safe and secure.
\item Finally, we assume that the number of registered peers is adequate for providing the required service. We precisely assume that there are at least two different available registered peers at any moment for each service request.
\end{itemize}


\vspace{-1mm}
\section{Timed-release service protocol}
\label{s3}
\vspace{-1mm}
We present the proposed timed-release service protocol organized along four subsections, each of which discusses a key component of the protocol. The notations used in this section are summarized in Table \ref{t1}.


\vspace{-1mm}
\subsection{Peer registration} 
\label{s3-1}
In this subsection, we present the first part of the protocol, \textit{peer registration}, designed for allowing peers to make themselves known to the network. After presenting the protocol, we discuss the peer working window and deposit management in more detail.
To set up timed-release services, a prerequisite is to have a platform for making peers $P$s and data senders $S$s know each other. 
Since peers and senders have no trust in each other, instead of a face-to-face negotiation, they need to transfer their information (peer working window $T^w$ and sender storage window $T^s$) and money (remuneration and deposit) to the decentralized smart contract $C$ and treat $C$ as a trusted intermediary to put the deal through. A new peer registers by sending their working windows, public keys and deposit to join the contract $C$. This information is recorded in the registration list maintained by $C$.

\begin{small}
\begin{mdframed}[innerleftmargin=1.6pt]
\centerline{\textbf{Peer registration protocol}}
\begin{packed_enum}[leftmargin=*]
  \item  To be registered, each peer must submit a set of future working windows $T^w$s and a public key to contract $C$. It must also pay a deposit to contract $C$ as assurance of no misbehavior while providing future services.
  \item  Each peer agrees to complete any assigned jobs. 
  \item  Each peer agrees to allow the contract to freeze a part of its deposit for an assigned job until the job is completed.
  \item  Each peer agrees to renew the public key for each job.
  \item  Each peer can modify working windows $T^w$s and the unfrozen deposit at any time, but jobs assigned before modification should still be completed.
 
\end{packed_enum}
\end{mdframed}
\end{small}
 
\noindent \textbf{Peer working windows}: 
As discussed in Section \ref{s2-2}, the proposed timed-release service protocol splits a long storage time duration, $T^s$  into a series of successive shorter time durations, each of which is handled by a different peer during its working window, $T^w$, as the encrypted secret key gets routed on the blockchain network. 
Figure~\ref{f3} shows an example representing $T^w$ as horizontal segments in a coordinate frame with timeline and peer indexes as $x$ and $y$ axes respectively. Here, the segment at the bottom-left corner represents a working window $[t_1,t_2]$ belonging to $P_i$. 


\begin{table}
\centering
\caption{Summary of notations}
\vspace{-2mm}
\begin{tabular}{ |L{0.1cm}|p{1.8cm}|p{5.5cm}| }
  \cline{2-3}
  \multicolumn{1}{c|}{} & 
  \textbf{Notations} & 
  \textbf{Descriptions} \\ 
  \cline{2-3}
  \hline
  \multirow{7}{*}{\rotatebox[origin=c]{90}{\parbox[c]{0.6cm}{\centering \textbf{Temporal}}}}
  & \cellcolor{gray!20} $t,T,|T|$ & \cellcolor{gray!20} Time point, time window, length of a window.\\
  & \multirow{2}{*}{$T^{w}=[t_b,t_e]$} & Peer working window, which begins at $t_b$ and ends at $t_e$.\\
  & \cellcolor{gray!20} \multirow{2}{*}{$T^s=[t_s,t_r]$} & \cellcolor{gray!20} Sender storage window, from setup time $t_s$ to release time $t_r$.\\
  & $|T_t|$ & Data transfer period.\\
  \hline
  \multirow{6}{*}{\rotatebox[origin=c]{90}{\parbox[c]{1.3cm}{\centering \textbf{Monetary}}}}
  & \cellcolor{gray!20} $v$ & \cellcolor{gray!20} Monetary value of private data sent to recipient. \\
  & $d^{a}$ & Balance of a deposit account. \\
  & \cellcolor{gray!20} $d^{s}$ & \cellcolor{gray!20} Deposit required by a service. \\
  &  \multirow{2}{*}{$r,\widehat{r}$} &  Remuneration paid by sender to one peer or all peers.\\
  & \cellcolor{gray!20} $p$ & \cellcolor{gray!20} Payment charged at the setup time. \\
  & $c$ & Cost of a selected peer during a service. \\
  & \cellcolor{gray!20} $a$ & \cellcolor{gray!20} Reporting award of a service. \\
  \hline
\end{tabular}
\vspace{-6mm}
\label{t1}
\end{table}

\noindent \textbf{Deposit management mechanism}:
The proposed protocol uses deposits as a mechanism to penalize peer misbehaviors in order to prevent drop and release-ahead attacks.
Senders may want to pay more for getting a higher deposit from peers as guarantees of their behaviors to send private data with higher monetary value $v$. 
To support such requirements, we design a dynamic deposit management mechanism that incorporates deposit with two states: frozen and unfrozen. 
One can imagine that each peer has a deposit account in contract $C$. The deposit account is opened after registration and its balance is denoted as $d^a$. Initially, $d^a$ is unfrozen. Later, data senders can calculate the amount of deposit they want from peers, denoted as $d^s$, based on the monetary value of the private data $v$. Then, during service setup, senders should only select peers from the registration list with at least $d^s$ unfrozen deposit. 
The amount of $d^s$ deposit, once being verified by contract $C$, will be frozen from accounts of selected peers until the end of their services.
At any time, each peer can only manage its unfrozen deposit in account as the ownership of the frozen part has been temporarily transferred to contract $C$.
In this way, the designed deposit management mechanism encourages peers with secure storage environment to keep a high deposit balance so that they can get jobs requiring a higher deposit $d^s$ to earn more payments by taking higher risk.


\vspace{-1mm}
\subsection{Service setup}
Next, we present the second part of the protocol, namely \textit{service setup}, designed for allowing senders to select peers from the registration list based on their requirements and set up the service with contract $C$ after paying remunerations. 
We first present the protocol description for service setup and then illustrate the remuneration computation and peer selection algorithm in detail.

\label{s3-2}
\begin{small}
\begin{mdframed}[innerleftmargin=1.6pt]
\centerline{\textbf{Service setup protocol}}
\begin{packed_enum}[leftmargin=*]
  \item Before setup time $t_s$, senders compute the remuneration $\widehat{r}$ and deposit $d^s$ required by this service and then locally run the peer selection algorithm to select peers from the registration list satisfying their requirements.
  \item  At setup time $t_s$, senders submit service information including selected peers to contract $C$. Also, both sender $S$ and recipient $R$ should pay $p>d^s+\widehat{r}$ to contract $C$.
  \item  Upon receiving a setup request, contract $C$ calculates remuneration $\widehat{r}$ and deposit $d^s$ of this service, then:
  \begin{packed_enum}
    \item  If $p>d^s+\widehat{r}$ and each selected peer has unfrozen deposit higher than $d^s$, $C$ will approve the setup, freeze $d^s$ of selected peers and refund $p-d^s-\widehat{r}$ to $S$, $p-d^s$ to $R$.
    \item  Otherwise, $C$ will reject the setup and refund $p$ to $S$, $R$.
  \end{packed_enum}
\end{packed_enum}
\end{mdframed}
\end{small}

\noindent \textbf{Remuneration computation}:
The total remuneration $\widehat{r}$ paid by the sender consists of two parts $\widehat{r_c}$ and $\widehat{r_s}$. 
The $\widehat{r_c}$ component is charged to compensate the cost of peers for invoking functions of contract $C$ during the service, so $\widehat{r_c}=k r_c$ for $k$ selected peers. 
The $\widehat{r_s}$ component is charged to reward peers for storing the secret key, so it should be higher for longer storage time $|T^s|$. 
Meanwhile, to encourage more peers to serve for long-term storage, senders should be charged more for a later storage hour closer to release time $t_r$ than an earlier one closer to setup time $t_s$.
Therefore, if we represent the charge of $i^{th}$ storage hour as $\triangle r_s^i$ and set the first hour charge as $\triangle r_s^1$, by setting per hour increment of $\triangle r_s^i$ as $\alpha$, we get $\triangle r_s^i=\triangle r_s^{i-1}+\alpha$, which further gives $\widehat{r_s}=\frac{|T^s|[\triangle r_s^1+\triangle r_s^1+(|T^s|-1)\alpha]}{2}=|T^s|\triangle r_s^1+\frac{|T^s|(|T^s|-1)}{2}\alpha$.
Additionally, $S$ should be charged more for a higher monetary value of private data $v$ as an incentive to make peers maintain higher balance in deposit accounts, so we consider the above $\widehat{r_c}$ and $\widehat{r_s}$ as the charging standard when $v=\triangle v$ (e.g., $\triangle v=\$100$) and adjust the final $\widehat{r}$ based on that.
To sum up, a sender should pay remuneration $\widehat{r}=(\lceil \frac{v}{\triangle v} \rceil)^{\beta}  [k r_c+|T^s|\triangle r_s^1+\frac{|T^s|(|T^s|-1)}{2}\alpha]$ in total and a peer serving for $i^{th}$ to $j^{th}$ hours in $T^s$ should be paid $r=(\lceil \frac{v}{\triangle v} \rceil)^{\beta} [r_c+(j-i)\triangle r_s^1+\frac{i+j-2}{2}\alpha]$, where $\alpha>0$ and $\beta>1$.

\noindent \textbf{Peer selection}:
The peer selection algorithm has two objectives, namely (i) minimizing remunerations paid by senders and (ii) maximizing the expected profit made by the peers. 
To realize the first objective, we note that the only way to reduce remuneration $\widehat{r}$ is to make $k$ smaller, namely selecting fewer peers for a service, which does not impact the expected profit $r$ earned by selected peers as $\widehat{r_s}$ is fixed. For achieving the second objective, we need the algorithm to always pick earlier hours in peer working windows $T^w$s first so that deposit $d^s$ can be unfrozen as soon as possible. For example, the algorithm needs to pick just one hour from a ten-hour $T^w$ for a one-hour service. By picking the last hour in $T^w$, that peer only receives an one-hour profit because deposit $d^s$ has to be frozen for the entire ten hours. In contrast, by picking the first hour, since deposit $d^s$ will be unfrozen after one hour, the door for accepting new jobs is reopened and that peer can make a ten-hour profit in the best case.
We design a greedy algorithm to achieve both of these objectives simultaneously. By decomposing the peer selection problem into a series of subproblems, we define each subproblem as `given all peer working windows $T^w$s covering an input time point, output the $T^w$ that makes the total number of selected peers minimum'. 
Once a $T^w$ is selected, its beginning time $t_b$ is then used as the input time point of the subsequent subproblem to select the next peer. Intuitively, in a subproblem, the greedy choice is to pick the $T^w$ with earliest $t_b$. 
Next, we demonstrate the peer selection process with an example in Figure \ref{f3}. The pseudo-code and demonstration of the peer selection algorithm can be found in Appendix~\ref{appendix_A}.

\begin{figure}
\centering
{
    \includegraphics[width=8cm,height=4cm]{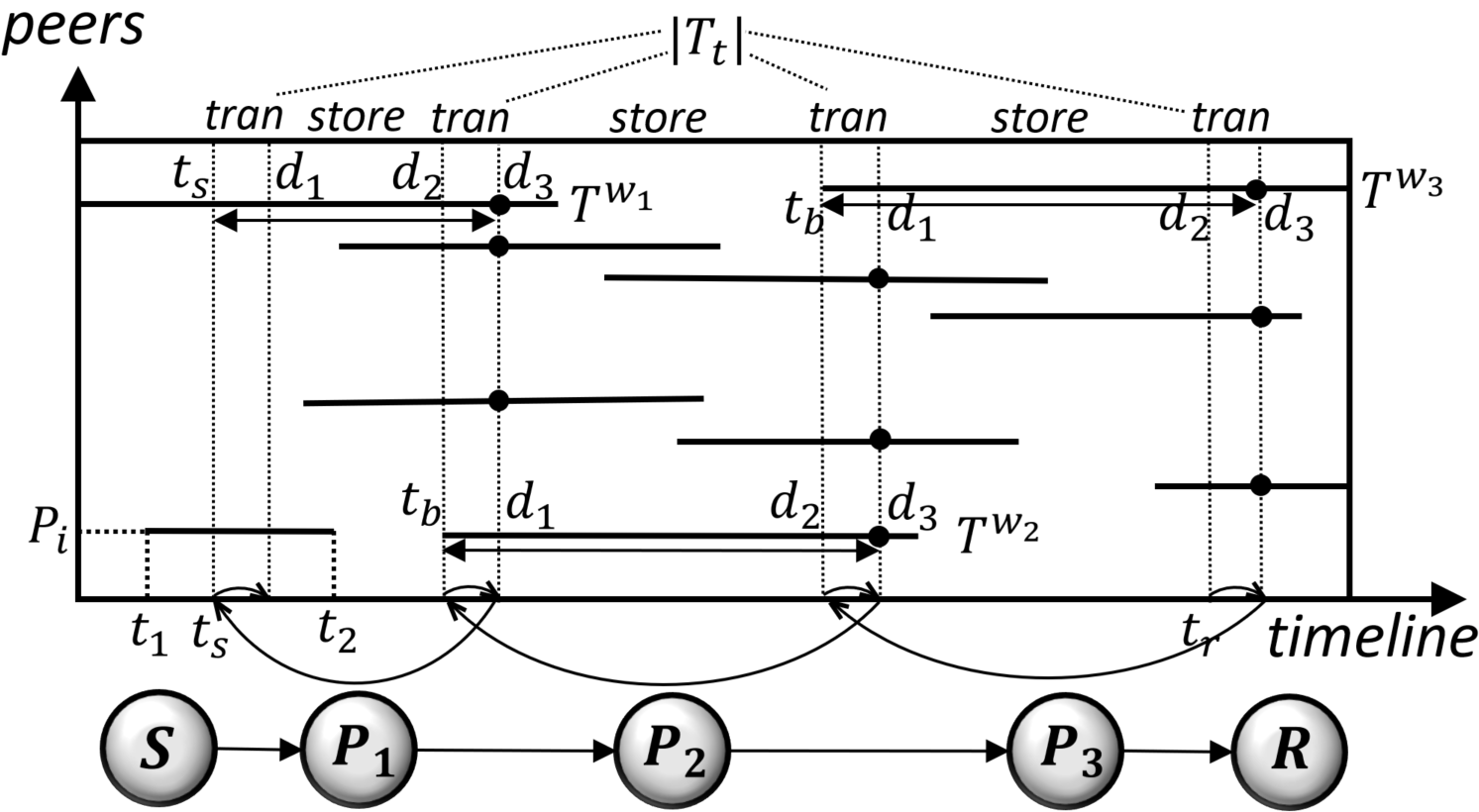}
}
\vspace{-2mm}
\caption {Peer selection}
\vspace{-7mm}
\label{f3} 
\end{figure}

In the example, instead of release time $t_r$, the algorithm takes $t_r+|T_t|$ as the input time point of the first-round subproblem as we need to leave a buffer zone $|T_t|$ for data transfer between each pair of adjacent peers on path. 
In the first round, there are three available peer working windows $T^w$s covering $t_r+|T_t|$ and obviously $T^{w_3}$, due to its earliest begin time $t_b$ among the three, is the greedy choice. 
As a result, we select $P_3$ as the last peer on path and set $T^{w_3}.t_b+|T_i|$ as input of the second-round subproblem. We then get $T^{w_2}$ as second-round greedy choice, so we select $P_2$ and set $T^{w_2}.t_b+|T_i|$ as input of the third-round subproblem, which gives $T^{w_1}$ as third-round greedy choice to pick $P_1$. This is the end of peer selection process as $T^{w_1}$ has already covered setup time.

\subsection{Service enforcement}
\label{s3-3}
The third component of the protocol deals with \textit{service enforcement} that specifies the behaviors that should be followed by the sender $S$, recipient $R$ and peers, $P$s during the service process to render the service successfully.
The protocol sets deadlines for each behavior and treats any missing behavior as a drop attack to enable drop attacks behavior to be detectable.
Next, we present the protocol with a discussion on the designed behaviors. We then model the protocol as an extensive-form game with imperfect information to prove that any rational participating peer will always follow the protocol honestly.


\begin{small}
\begin{mdframed}[innerleftmargin=1.6pt]
\centerline{\textbf{Service enforcement protocol}}
\begin{packed_enum}[leftmargin=*]
  \item  Before time $t_s+|T_t|$, the sender must submit hashes of certificates, hash of the secret key cyphertext (only encrypted by the recipient's public key) and finally the encrypted whisper key to contract $C$. It must also encrypt the secret key using public keys of selected peers and transfer it to the first peer.
  \item Each selected peer must decrypt one layer of the received encrypted secret key, submit the obtained certificate to contract $C$ and verify the behavior of previous participants before its first deadline $d_1$. It must submit encrypted whisper key to contract $C$ before its second deadline $d_2$ and transfer the secret key to the next peer before its third deadline $d_3$.
  \item  Before time $t_r+|T_t|$, the recipient must first decrypt the last layer of the encrypted secret key to submit the obtained certificate to contract $C$ and then verify the behavior of both the previous participants and the recipient itself.
  \item  If any verification launched by a peer (or recipient) in term 2 (or 3) gives $False$, $C$ should immediately terminate the service and judge the last participant on the path that fails to pass the verification to be guilty. Then, $C$ should refund deposit $d^s$ to all innocent participants, pay remuneration $r$ to innocent peers and issue confiscated $d^s$ and unused $r$ to sender.
  \item  If a verification gives $Ture$, contract $C$ should refund deposit $d^s$ and pay remuneration $r$ to all participants that have already honestly finished their job before their deadlines.
\end{packed_enum}
\end{mdframed}
\end{small}

\begin{figure*}
\vspace{-5mm}
\centering
{
    \includegraphics[width=13cm,height=3cm]{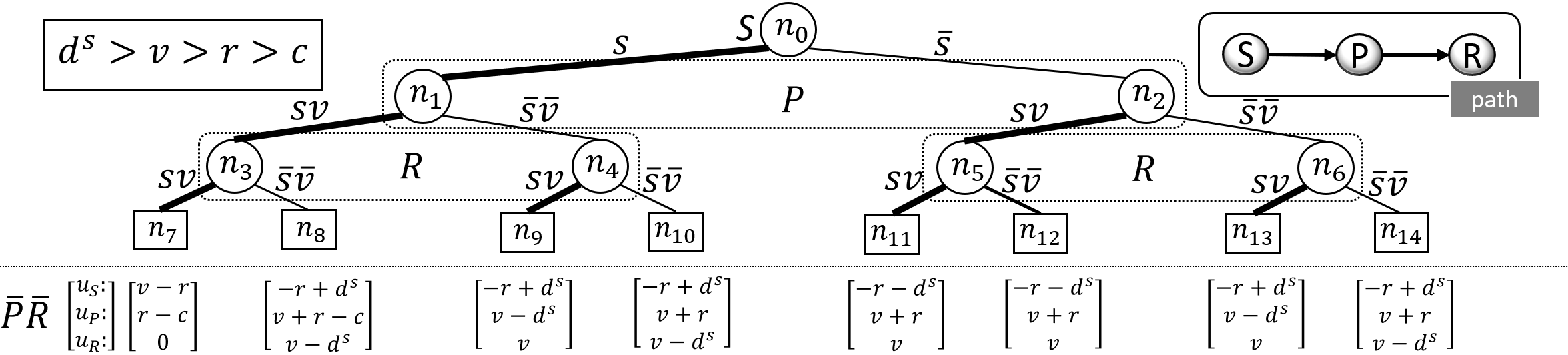}
}
\vspace{-2mm}
\caption {Game tree induced by service enforcement protocol}
\vspace{-5mm}
\label{f4} 
\end{figure*}

\vspace{1mm}
\noindent \textbf{Whisper key submission}:
Our system employs the Whisper protocol~\cite{Whisper2017} to transfer secret keys between any two Ethereum peers by building private channels with symmetrical whisper keys. 
Specifically, the first peer should encrypt its whisper key with the public key of the second peer and submit it to contract $C$ so that only the second peer can get the whisper key and set up the channel.

\vspace{1mm}
\noindent \textbf{Certificate}:
We design certificates for detecting drop attacks. For each peer and recipient, we need the sender to secretly generate a unique certificate and package it along with the corresponding layer of the encrypted secret key. Therefore, upon decrypting the received encrypted secret key with the private key, the peer (or recipient) will get the unique certificate.
The peer (or recipient) then should submit the certificate to contract $C$. If the hash of the submitted certificate is same as the one submitted by sender, the correct reception of encrypted secret key can be proved. Otherwise, a drop attack is detected. 
However, with certificates, we can only detect that a drop attack has happened between two adjacent peers. It is hard to further figure out which of the two peers launched the attack as the channel between them is private. We will discuss how to handle such a dispute in Section \ref{s3-4}. 

\noindent \textbf{Verification}:
We design verification as a function of contract $C$ for enforcing submission of whisper keys and certificates. A missing whisper key or certificate, both causing a drop attack, cannot be automatically detected by contract $C$. Here, we need the verification function to be triggered by Ethereum peers to check whether the submissions have been made on time. If all the submissions have been correctly made until the time of verification, the function returns a $True$. Otherwise, it returns a $False$.
For each timed-release service, multiple verifications are required to detect a drop attack in a timely manner so that the service can be terminated on time and deposits of innocent peers can be unfrozen quickly.
We carefully design the protocol as an extensive-form game with imperfect information to prove that any rational participant in this game will always choose to submit both whisper key and certificate on time. 

\noindent \textbf{The game induced by the protocol}:
We model the protocol as an \textit{extensive-form game with imperfect information}~\cite{leyton2008essentials}, which can be represented as a game tree in Figure \ref{f4}.
For ease of explanation, the example only has one peer $P$ between sender $S$ and recipient $R$ on path, but the services with more peers follow the same result.
The game has three players $\{S,P,R\}$. Its basic actions are (whisper key and/or certificate) submission  ($s$) and verification ($v$), so the action set is $\{s,v,\bar{s},\bar{v},sv,s\bar{v},\bar{s}v,\bar{s}\bar{v}\}$, where $\bar{s}$ and $\bar{v}$ represent no submission and no verification respectively and $sv,s\bar{v},\bar{s}v,\bar{s}\bar{v}$ stand for the combinations. 
The game tree consists of choice nodes $\{n_0,...,n_{14}\}$ and terminal nodes $\{n_{15},...,n_{30}\}$.
At the beginning of the game, sender $S$ ($\{n_0\}$) can choose either to submit whisper key or not by taking one action from $\{s,\bar{s}\}$.
Then, the game moves to peer $P$ ($\{n_1,n_2\}$), who has no idea about the choice made by sender $S$ (imperfect information). 
The peer $P$ should choose one action from $\{sv,s\bar{v},\bar{s}v,\bar{s}\bar{v}\}$, namely four combinations of doing submission and verification or not, but we argue that $s\bar{v}$ and $\bar{s}v$ can be omitted.
The reason is that a peer $P$ choosing $s\bar{v}$ gets same payoff as choosing $sv$ if no previous player has chosen $\bar{s}$ as there is no need of verification in this situation. In contrast, if there is at least one previous player who has chosen $\bar{s}$, the payoff by choosing $s\bar{v}$ is equal to that of choosing $\bar{s}\bar{v}$ as there is no need of submission when a drop attack has been launched earlier. As a result, $s\bar{v}$ can be replaced by $sv$ and $\bar{s}\bar{v}$, and it is also true for $\bar{s}v$.
Finally, the game goes to the turn of recipient $R$ ($\{n_3,n_4\},\{n_5,n_6\}$), who has no idea of the action taken by sender $S$ and peer $P$. Similar to $P$, recipient $R$ should choose one action from $\{sv,\bar{s}\bar{v}\}$, but $s$ here only means the certificate submission as it is the last peer on the path. 

We now analyze the payoffs shown under the terminal nodes, where $u_S$, $u_P$ and $u_R$ represent payoff of sender $S$, peer $P$ and recipient $R$ respectively. 
The payoffs have uncertainty. Most peers on the path, by dropping the encrypted secret key, can only save a service cost $c$, but some peers can get an additional profit no more than the monetary value of the private data $v$ (for ease of presentation, we represent it as $v$ in this game).
Therefore, it is uncertain whether peer $P$ and receipt $R$ can get the additional benefit $v$ from dropping the package. 
To model this uncertainty, we use $P$ and $R$ to represent the ones only targeting at $c$ and $\bar{P}$ and $\bar{R}$ to represent the ones also targeting at $v$.
By considering this uncertainty, this game can be modeled as a more sophisticated Bayesian game~\cite{armbruster1979bayesian}. However, we find that the four situations in this game ($\{PR,P\bar{R},\bar{P}R,\bar{P}\bar{R}\}$) can reach the same Nash equilibrium~\cite{nash1950equilibrium} and therefore, for ease of explanation, we will only analyze the situation that both peer $P$ and recipient $R$ can get additional benefit $v$, namely $\bar{P}\bar{R}$.

In $\bar{P}\bar{R}$, we will show that if deposit $d^s>v$ is satisfied, then the best choice of each player is to do both submission and verification on time. 
We start from analyzing the choice of recipient $R$ between $sv$ and $\bar{s}\bar{v}$ at the last step of this game. 
At $n_3$, by choosing $sv$, $R$ gets $0$ at $n_{7}$, which is higher than $u_R=v-d^s$ at $n_{8}$ if $\bar{s}\bar{v}$ is chosen and $d^s>v$ is satisfied. By further checking $n_4$ to $n_6$, we can find $sv$ always brings $u_R$ no less than $u_R$ from $\bar{s}\bar{v}$, which proves that $sv$ dominates $\bar{s}\bar{v}$ and $R$ should always choose $sv$ no matter how the game has been played before.
Following the same rule, peer $P$ should always choose $sv$ at $\{n_1,n_2\}$ if $d^s>v-(r-c)$ is satisfied. Since we need $r>c$ to make $P$s get positive profit from the service, $d^s>v-(r-c)$ can be automatically satisfied when $d^s>v$.
Finally, with the same rule, sender $S$ should always choose $s$ at $n_0$.

In game theory, if by taking a strategy, a player can make the expected payoff no less than that induced by taking any other strategy no matter what strategies are taken by other players, this strategy will become his or her best response. If all the players are taking their best responses, the game will reach a Nash equilibrium~\cite{nash1950equilibrium}. Nash equilibrium is the most important solution concept in game theory, which describes a situation that every player chooses the best response and no one can make payoff higher by changing strategy if no one else changes strategy. In this game, the Nash equilibrium is reached when all the players follow the bold edges, which results in all rational players, whether they are sender, recipient or peers, choosing to honestly obey the protocol.
 
\subsection{Reporting mechanism}
\label{s3-4}
In this subsection, we present the last part of the protocol, namely \textit{reporting}, designed for handling both release-ahead attacks and the dispute of drop attacks that are hard to be detected by \textit{service enforcement} protocol. 

\begin{small}
\begin{mdframed}[innerleftmargin=1.6pt]
\centerline{\textbf{Reporting protocol}}
\begin{packed_enum}[leftmargin=*]
  \item  Any peer can report a release-ahead attack to contract $C$ with evidence before $t_r$. If the evidence is true, contract $C$ should judge the suspect to be guilty, confiscate its deposit $d^s$ and terminate the service. Then, contract $C$ should refund deposit $d^s$ to all innocent participants, pay remuneration $r$ to all innocent peers, and finally split deposit $d^s$ of the suspect to an award $a$ paying to the reporter and the rest $d_p-a$ paying to the sender.
  \item  Any peer on the path can report a dispute of drop attack between a suspect (the peer before this reporter on path) and the reporter to contract $C$ before deposit $d^s$ of the suspect is unfrozen. Upon receiving the report, contract $C$ should terminate the service and confiscate $d^s$ of $S$, suspect and reporter. Then, contract $C$ should refund $d^s$ to all innocent participants and pay remuneration $r$ to all innocent peers. The confiscated $3d^s$ should be split to an award $a$ paying to the reporter and the rest $3d_p-a$ locked in contract $C$ without being able to be withdrawn by anyone including the contract designer.
  
\end{packed_enum}
\end{mdframed}
\end{small}

\vspace{1mm}
\noindent \textbf{Release-ahead attack}:
As discussed in Section \ref{s2-3}, it is highly difficult to detect a secret release attack made by peers on the path. We design a reporting mechanism to enable a release-ahead attack to be reported with evidence by adversaries themselves. 
The evidence should be the secret key cyphertext (only encrypted by recipient's public key), which is a necessity for any adversary (including the recipient) to obtain the secret key plaintext because its correctness can be checked with the hash value stored in blockchain, thus protecting the adversary from being cheated by the untrusted bribee.
The secret key cyphertext (only encrypted by recipient's public key) uploaded before $t_r$ proves that the last selected peer on path has released the data. 
Then, contract $C$ can verify the correctness of the evidence with the hash value stored in blockchain. 
If the evidence is proved to be true, the adversary will get an award, $a$ from contract $C$ while the last selected peer will lose its deposit $d^s$. 
It may sound irrational not to penalize but reward the adversary. 
However, this anti-intuitive reporting mechanism is an effective way to prevent release-ahead attacks as long as both adversary and the peer are rational. In the game between them, the best response of the adversary is to always report the peer to earn the award $a$ from contract $C$ without any penalty. Based on this knowledge, the best response of any peer on the path is to never accept bribery. Therefore, the Nash equilibrium of this game makes such a release-ahead attack never happen.

\noindent \textbf{Dispute of drop attack}:
As discussed in Section \ref{s3-3}, drop attacks cannot be solely prevented by verifications. 
After a drop attack is detected between two adjacent peers on the path when the second peer between the two fails to submit the correct certificate, it is hard to figure out which peer actually launched it. It can be either launched by the first peer by not sending the correct encrypted secret key to the second peer or by the second peer by maliciously denying the reception of the encrypted secret key. In addition, it can be launched by the sender $S$ by submitting fake hashes of certificates to contract $C$ at the very beginning. 
To solve it, we allow the second account to report the dispute. Upon receiving the report, contract $C$ should confiscate deposit $d^s$ of the three participants and send back an award $a$ to the second peer.
Again, this anti-intuitive reporting mechanism is an effective way to prevent drop attack dispute by making the three participants as a community of interests as long as these accounts are rational.
In this game, when there is a drop attack, the second peer has the dominant action to always report the dispute because it will lose part of its deposit $d^s-a$ by reporting it but lose the entire deposit $d^s$ due to the missing certificate by not reporting it. 
With this knowledge, the best response of the first peer and sender is to never launch a drop attack because otherwise they will lose the entire deposit $d^s>v$ due to the report. 
Finally, given the best response of the first peer and sender, if $d^s>v+a$ is satisfied, the best response of the second peer is also to never launch a drop attack because otherwise it will lose $d^s-a>v$ due to the report. 
As a result, the Nash equilibrium is reached when all of them choose to never launch a drop attack.

\begin{table}
\centering
\begin{tabular}{c c p{1.5cm} p{3.5cm}} \toprule 
    {\textbf{Sections}} & {\textbf{Invokers}} & {\textbf{Functions}} & {\textbf{Purposes}} \\ \midrule
    \multirow{4}{*}{\rotatebox[origin=c]{90}{\parbox[c]{1.1cm}{\centering \textbf{Register}}}}
    & P & newPeer & register a new Peer \\
    & P & updateBalance  & update deposit balance \\
    & P & updateWindow  & update working windows \\
    & P & updatePubKey  & update public keys \\ \midrule
    \multirow{3}{*}{\rotatebox[origin=c]{90}{\parbox[c]{0.8cm}{\centering \textbf{Setup}}}}
    & S & senderSign   & sender signs the contract\\
    & R & recipientSign  & recipient signs the contract \\
    & S & setup  & setup the service \\ \midrule
    \multirow{4}{*}{\rotatebox[origin=c]{90}{\parbox[c]{1.1cm}{\centering \textbf{Enforce}}}}
    & S & setCert  & submit hashes of certificates \\ 
    & P,R & verifyCert  & verify received certificates \\
    & P & setWhisperKey   & submit encrypted whisper keys \\
    & P,R & verification  & do verification \\ \midrule
    \multirow{4}{*}{\rotatebox[origin=c]{90}{\parbox[c]{1.1cm}{\centering \textbf{Report}}}}
    & Any & releaseReport  & report a release-ahead attack \\ 
    & Any & releaseAward  & get award for reporting release\\
    & P,R & dropReport   & report a drop attack \\
    & P,R & dropAward  & get award for reporting drop\\ \bottomrule
\end{tabular}
\vspace{-1mm}
\caption{Summary of functions in the smart contract}
\vspace{-7mm}
\label{t2}
\end{table}

\begin{figure*}
\centering
\subfigure[{\small All windows}]
{
   \label{SC_r11}
   \includegraphics[width=0.42\columnwidth]{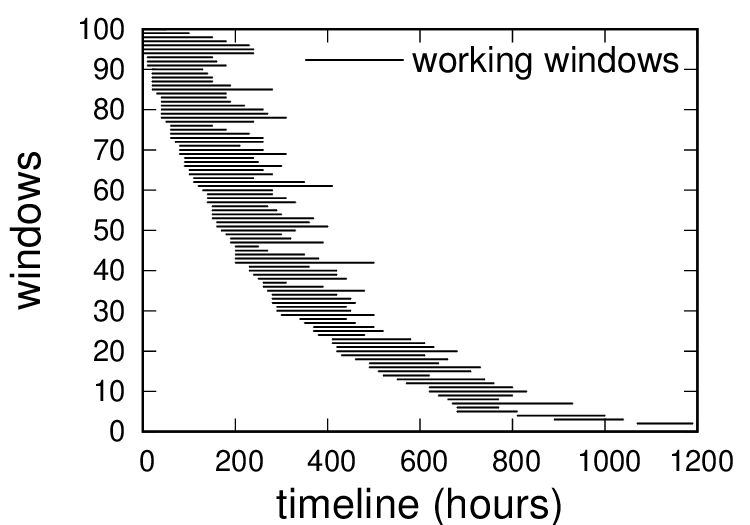}
}
\subfigure[{\small Selected windows (300h)}]
{
	\label{SC_r12}
    \includegraphics[width=0.42\columnwidth]{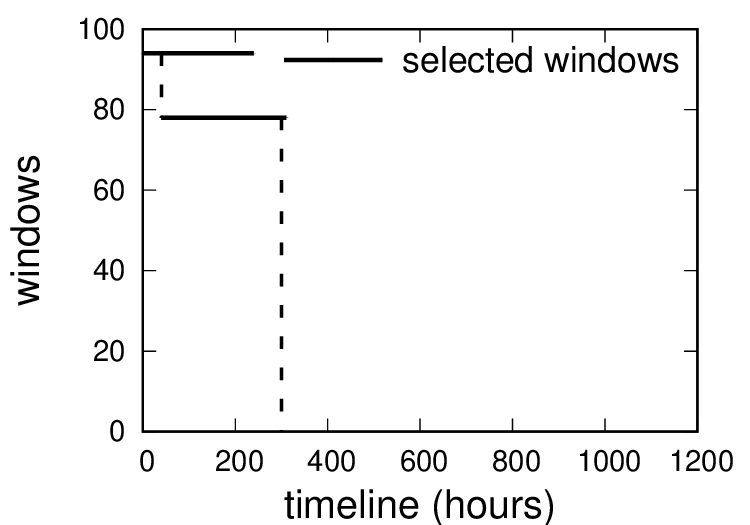}
}
\subfigure[{\small Selected windows (600h)}]
{
   \label{SC_r13}
   \includegraphics[width=0.42\columnwidth]{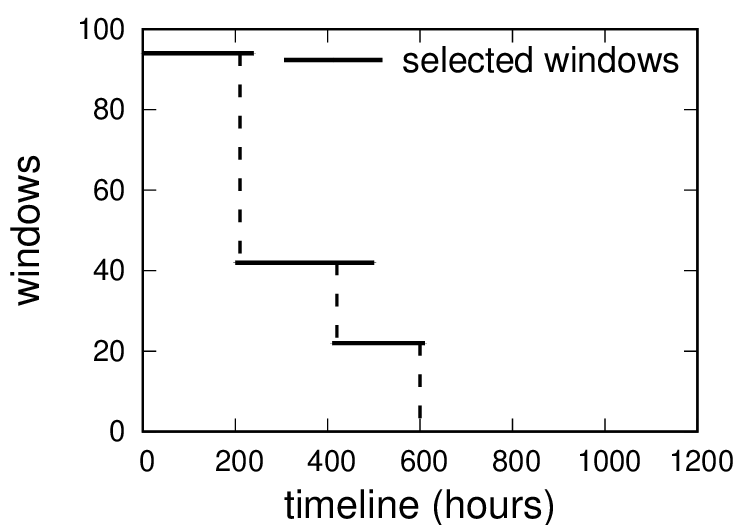}
}
\subfigure[{\small Selected windows (1000h)}]
{
	\label{SC_r14}
    \includegraphics[width=0.42\columnwidth]{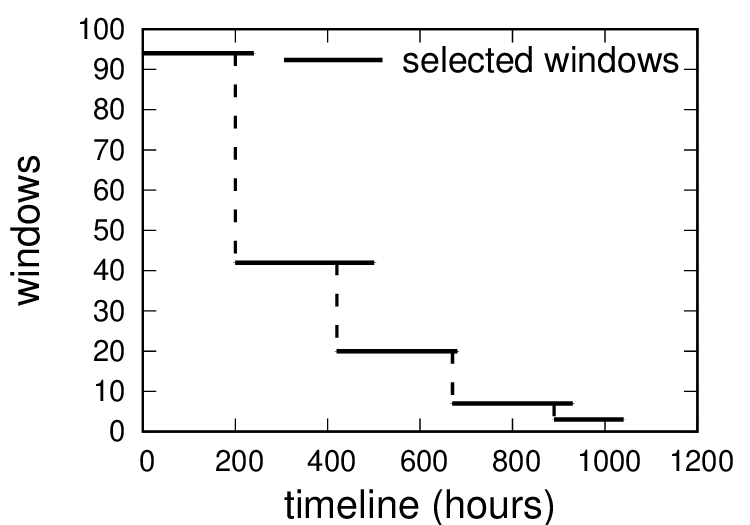}
}
\vspace{-3.5mm}
\caption{Peer selection}
\label {r1}
\vspace{-5mm}
\end{figure*}

\begin{figure*}
\centering
\subfigure[{\small Monetary cost (3 peers)}]
{
   \label{SC_r21}
   \includegraphics[width=0.42\columnwidth]{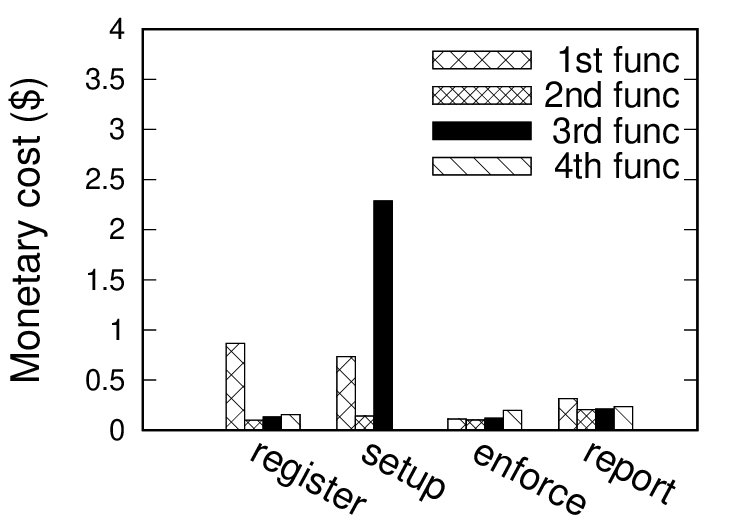}
}
\subfigure[{\small Monetary cost (5 peers)}]
{
	\label{SC_r22}
    \includegraphics[width=0.42\columnwidth]{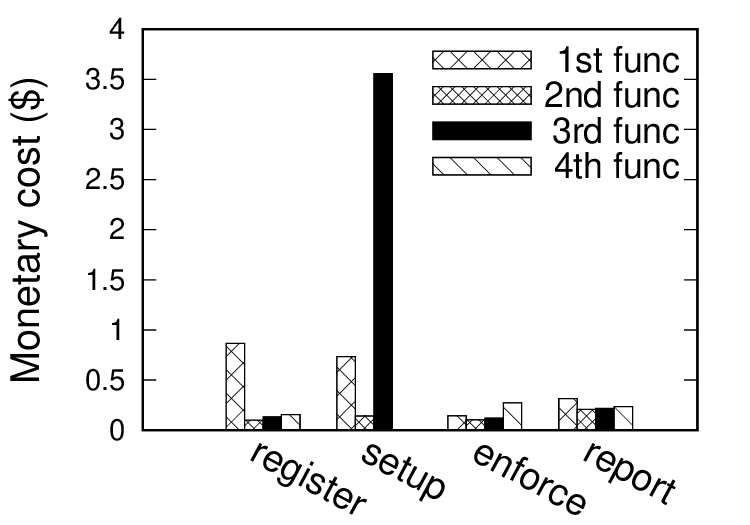}
}
\subfigure[{\small Time overhead (3 peers)}]
{
   \label{SC_r23}
   \includegraphics[width=0.42\columnwidth]{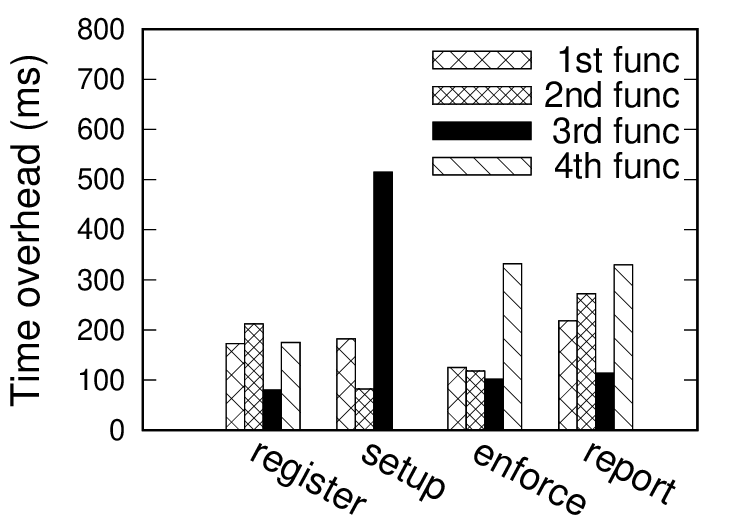}
}
\subfigure[{\small Time overhead (5 peers)}]
{
	\label{SC_r24}
    \includegraphics[width=0.42\columnwidth]{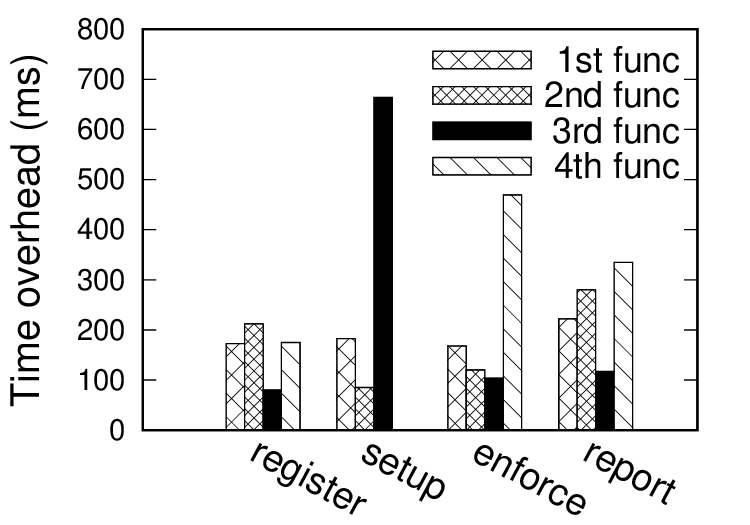}
}

\vspace{-3.5mm}
\caption{Performance evaluation}
\label {r2}
\vspace{-7mm}
\end{figure*}

\vspace{-1mm}
\section{Implementation} 
\vspace{-2mm}
\label{s4}
In this section, we present the implementation of the proposed self-emerging data release smart contract and experimentally evaluate its performance and security.

\subsection{Implementation}
We first introduce the implementation setup and then present the functions created in the smart contract and demonstrate how they work in practice. Finally, we present the test instance for our experimental evaluation.

\noindent \textbf{Setup}:
We programmed the smart contract in the contract-oriented programming language \textit{Solidity}~\cite{Solidity2017}, deployed it to the Ethereum official test network \textit{rinkeby}~\cite{Rinkeby2017} and tested it with Ethereum official Go implementation \textit{Geth}~\cite{Geth2017}. We used the \textit{SolRsaVerify} contract~\cite{SolRsaVerify2017} to verify signatures in the release reporting mechanism. We ran our experiments on an Intel Core i7 2.70GHz PC with 16GB RAM.

\noindent \textbf{Contract functions}:
We design the smart contract to include 15 main functions for supporting the four parts of protocol presented in Section \ref{s3}. The functions are shown in Table~\ref{t2} with their respective invokers and purposes. For example, function \textit{newPeer()} is designed to be invoked by peers during registration phase for being registered into the list.

\begin{itemize}[leftmargin=*]
\item \textbf{Registration}: Peers ($P$s) can first invoke \textit{newPeer()} to be registered and recorded into the peer list and then manage their unfrozen deposit balance, working windows and public keys through the other three functions.
\item \textbf{Setup}: A sender ($S$) should download the peer list, locally run peer selection algorithm to select peers from the list and estimate remuneration $r$. Then, $S$ should sign the contract through \textit{senderSign()} and also inform the recipient ($R$) to sign it through \textit{recipientSign()}. Finally, S should invoke \textit{setup()} to complete service setup and the smart contract ($C$) will freeze deposit $d^s$ of each selected $P$ after verifying payments of $S$ and $R$ and record the service information into the service list.
\item \textbf{Enforce}: At the beginning of a service, $S$ should invoke \textit{setCert()} to submit hashes of certificates to $C$. Then, during the service process, \textit{verifyCert()} is invoked by $P$s and $R$ to submit certifications, \textit{setWhisperKey()} is invoked by $P$s to submit encrypted whisper key and \textit{verification()} is invoked by $P$s and $R$ to do verification. 
\item \textbf{Report}: Any Ethereum peer can invoke \textit{releaseReport()} to report a release-ahead attack and get award through \textit{releaseAward()} after the report has been verified to be correct. Similarly, $P$ and $R$ on path can report a drop attack through \textit{dropReport()} and get award through \textit{dropAward()}. 
\end{itemize}

\noindent \textbf{Test instance}:
For testing purpose, we generated 100 Ethereum accounts to be registered as peers. Each peer offers one working window represented as a horizontal segment in Figure \ref{SC_r11}. We design an input parameter $Time$ to simulate the time during testing. 
As can be seen, the 100 working windows are distributed in the future 1200 hours. Their start times follow an exponential distribution with a mean of 300 hours while their lengths follow a normal distribution with a mean of 15 hours and a standard deviation of 5 hours. The reason is that we believe more peers may want to serve in the nearer future due to its lower uncertainty. From Figure \ref{SC_r12} to \ref{SC_r14}, we show the results of peer selection algorithm for sending the private data to 300, 600 and 1000 hours in the future by selecting two, three and five peers respectively. The storage on each selected $P$, upon hitting the dotted line, will be transferred to the next $P$. In all cases, storage on each $P$ starts from the beginning of its window, which signifies the design goal of the peer selection algorithm.

\subsection{Experimental evaluation}
We use the presented test instance to experimentally evaluate the performance and security of the smart contract. We begin by first evaluating the monetary cost and time overhead of the functions and then test the contract in different conditions including drop attack and release-ahead attack scenarios. 

\noindent \textbf{Monetary cost}:
The monetary costs of functions in Table~\ref{t2} for the three-peer case in Figure~\ref{SC_r13} are shown in Figure~\ref{SC_r21}. The results shown represent the maximum possible monetary costs for invoking the functions.
For ease of presentation, results are grouped into four clusters. Each cluster represents a protocol subsection and contains three or four functions following their order in Table~\ref{t2}.
In Ethereum, each function call will cost some gases if it changes the state of contract. Therefore, the raw data measured here is the gas cost of each function, which is then transferred to cost in \$ based on 1 gas = $1.0371979124 \times 10^{-8}$ ETH and 1 ETH = \$300 as of date, 10/29/2017~\cite{etherscan}.
As can be seen, most functions cost very little. Specifically, among the fifteen functions, eight cost lower than \$0.2 and twelve cost lower than \$0.3. The remaining three functions are \textit{newPeer()} (\$0.86), \textit{senderSign()} (\$0.73) and \textit{setup()} (\$2.29). They cost higher as data is stored into the registration list and service list in $C$ through the three functions. However, since each $P$ only calls \textit{newPeer()} for once during registration and each $S$ only calls \textit{senderSign()} and \textit{setup()} once during service setup, these costs are quite acceptable in practice. Thus, in case of three selected $P$s, a timed-release service costs \$5.07 in total, including \$3.33 cost incurred by $S$, \$0.44 cost incurred by each $P$ and \$0.41 cost incurred by $R$.
To study the scalability of the self-emerging data smart contract, we measured the monetary costs of the functions for the five-peer case in Figure~\ref{SC_r14} as Figure~\ref{SC_r22}. Compared with Figure~\ref{SC_r21}, only costs of three functions \textit{setup()}, \textit{setCert()} and \textit{verification()} are increased as a higher number of selected $P$s requires more data to be stored in data list with more certificates and more rounds of verifications. However, the increments of \textit{setCert()} and \textit{verification()} are quite small and the increment of \textit{setup()} from \$2.29 to \$3.56 is not a drastic overhead for storing the private data for a longer duration of 1000 hours.

\noindent \textbf{Time overhead}:
The time overheads of functions in Table \ref{t2} for the three-peer case in Figure~\ref{SC_r13} are shown in Figure~\ref{SC_r23}. All results are averaged for 100 tests. Among the fifteen functions, nine spent 0-200ms, three spent 200-300ms and two spent 300-400ms. The \textit{setup()} function spent the maximum time of 515ms due to the large amount of service data for storing. 
Again, we tested the five-peer case in Figure~\ref{SC_r14} and showed the results in Figure~\ref{SC_r24}. The two more selected peers make the same three functions \textit{setup()}, \textit{setCert()} and \textit{verification()} spend more time for the same reasons. Here again, the increments are quite acceptable.
In addition, we tested the time overhead of the peer selection algorithm, which shows that the algorithm is quite efficient by spending less than 20ms for even a peer list with 1000 working windows.

\noindent \textbf{Security evaluation}:
Finally, we evaluate the security protection offered by the smart contract by testing the results of a timed-release service in different conditions when the $S$, $R$ and $P$s engage in suspicious behaviors, shown in Table~\ref{t3}.
The test is based on the five-peer case in Figure~\ref{SC_r14} and the parameters about remuneration are set as $\alpha=0.000012$ ETH, $\beta=1.1$, $\triangle r_s^1=0.000001$ ETH, $\triangle v=1$ ETH, $r_t=0.002$ ETH respectively. The parameter setting can be adjusted, but it should not make the remuneration too low as in that case, one may not be incentivized to freeze \$1000 for half a year for earning a meager payment of \$0.1. In addition, we set $d^s=1.2v$ and award $a=0.1v$.
\vspace{-0.4mm}
\begin{itemize}[leftmargin=*]
\item Condition 1: Before the service, $S$, $R$ and the five $P$s all hold 5 available ETH. Then, $S$ wants to send a secret key with its monetary value $v=3$ ETH.
\item Condition 2: If all the participants follow the protocol honestly, $S$ can earn 2.872 ETH from the 3 ETH $v$ after paying 0.128 ETH to $P$s . Each $P$ can earn its remuneration based on the length of its service time as well as the distance of its service from the setup time $t_r$. As can be seen, the $P_5$ offering service for 890h-1000h earns much more than $P_1$ serving for the first 240 hours.
\item Condition 3: If $P_2$ does not submit its whisper key or certificate on time, its confiscated deposit $d_s=3.6$ ETH will make its final payoff to be $5-3.6+3=4.4$ ETH.
\item Condition 4: If $P_5$ releases its data to $P_2$, $P_2$ can report it to earn the 0.3 ETH award, which will make $P_5$ get $5-3.6+3=4.4$ ETH payoff.
\item Condition 5: If $P_4$ does not send the secret key to $P_5$ through the private channel, $P_5$ can report this drop dispute, which will make $P_4$ get 4.4 ETH payoff. Without reporting it to earn the 0.3 ETH award, $P_5$ can only get $5-3.6=1.4$ ETH payoff due to the failure of certificate submission.

\end{itemize}
\vspace{-0.4mm}
As can be seen, in conditions 3 to 5, adversaries with misbehavior only get 4.4 ETH payoff, which makes them lose 0.6 ETH. Therefore, a rational Ethereum peer should always choose to honestly follow the protocol resulting in condition~2.

\vspace{-1mm}
\section{Related work}
\vspace{-1.5mm}
\label{s5}
The problem of revealing private data only after a certain time in future has been researched for more than two decades. The problem was first described by May as timed-release cryptography in 1992~\cite{may1992timed} and has intrigued many researchers in the field of cryptography since then. 
One set of existing cryptographic solutions relies on a third party, also known as a time server, to release the protected information at the release time in future. The information, sometimes called time trapdoors, can be used by recipients to decrypt the encrypted message~\cite{blake2004scalable,chalkias2007improved,kasamatsu2012time,kikuchi2011strong,mont2003hp,rivest1996time}. 
Although efficiency and flexibility of time-server-based approaches have constantly been improved, the time server in this model has to be trusted not to collude with recipients so that encrypted messages cannot be entered before release time. This restriction makes this set of solutions involve a single point of trust.
Another set of existing solutions were designed to make data recipients solve a mathematical puzzle, called time-lock puzzle, before reading the messages~\cite{bitansky2016time,boneh2000timed,rivest1996time}. Such cryptographic solutions suffer from two key drawbacks. First, due to increasing advancements in computing hardware and hardware performance, the time taken by such puzzle computation is not determinate and hence these solutions cannot tackle the situations that demand the data be released with a precise release time. Secondly, the puzzle computation is associated with a significant computation cost. Incurring such high computation costs for a large big data infrastructure, such as the one considered in this work, does not lead to a scalable cost-effective solution.

\begin{table}
\centering
\begin{tabular}{c ccccccc} \toprule
    {\textbf{Cond}} & {\textbf{S}} & {\textbf{P1}} & {\textbf{P2}} & {\textbf{P3}} & {\textbf{P4}} & {\textbf{P5}} & {\textbf{R}} \\ \midrule
    1. & 5 & 5 & 5 & 5 & 5 & 5 & 5 \\
    2. & 7.872 & 5.010 & 5.017 & 5.026 & 5.035 & 5.040 & 5 \\
    3. & 8.489 & 5.010 &\cellcolor{gray!20} 4.4 & 5.026 & 5.035 & 5.040 & 5 \\
    4. & 8.212 & 5.310 & 5.017 & 5.026 & 5.035 & \cellcolor{gray!20} 4.4 & 5 \\ 
    5. & 1.347 & 5.010 & 5.017 & 5.026 &\cellcolor{gray!20} 4.4 & 1.7 & 5 \\ \bottomrule
\end{tabular}
\vspace{-2mm}
\caption{Security evaluation}      
\vspace{-8mm}
\label{t3}
\end{table}
 
In another direction of cryptographic solutions using blockchains~\cite{nakamoto2008bitcoin}, the difficulty of PoW (proof-of-work) can be diversely adjusted to change the average generation time of each block to a desired value, which makes blockchain to be a reference time clock with correctness guaranteed by the distributed network. Therefore, by combining witness encryption~\cite{garg2013witness} with blockchain~\cite{jager2015build,liu2015time}, one can leverage the computation power of PoW in blockchain to decrypt a message after a certain number of new blocks have been generated. However, the current implementation of witness encryption is far from practical, which requires an astronomical decryption time estimated to be $2^{100}$ seconds~\cite{liu2015time}.

Our preliminary work on decentralized self-emerging data has studied the problem in the context of Distributed Hash Table (DHT) networks~\cite{li2017timed}. The idea behind these techniques is to leverage the scalability and distributed features of DHT P2P networks to make message securely hidden before release time. In contrast to such DHT-based solutions that do not offer guaranteed resilience to potential misbehaviors, the decentralized self-emerging data release techniques presented in this paper employs a blockchain infrastructure that offers more robust and attractive features including higher protocol enforceability by using incentives and security deposits.

\vspace{-1.5mm}
\section{Conclusion}
\vspace{-1.3mm}
\label{s6}
In this paper, we develop decentralized techniques for supporting self-emerging data using smart contracts in Ethereum blockchain networks. 
Our proposed timed release service protocol implemented as a smart contract is nearly immutable in the Ethereum blockchain. The credibility and enforceability of the protocol are guaranteed through a careful design based on extensive-form games with imperfect information to prevent possible post-facto misbehaviors including drop attacks and release-ahead attacks. We developed the smart contract using \textit{Solidity} and implemented the system on the Ethereum official test network.
Our rigorous theoretical analysis and extensive experiments demonstrate the low monetary cost and the low time overhead associated with the proposed approach and validate its security properties. 
In future work, we plan to deal with the situation that powerful adversaries can make their controlled peers register even before the registration deadline selected by the data owners. 
Potential solutions include establishing a reputation system to make it harder for malicious peers to be selected or adopting secret share scheme~\cite{shamir1979share} to transmit shares of a secret key through multiple paths to make it harder for adversaries to restore the secret key.


\renewcommand\refname{Reference}
\bibliographystyle{plain}
\urlstyle{same}
\bibliography{main.bib}

\appendices

\section{}
\label{appendix_A}
We first present the pseudo-code of the peer selection algorithm and then demonstrate its feature.

\begin{algorithm}
\begin{footnotesize}
    \SetKwInOut{Input}{Input}
    \SetKwInOut{Output}{Output}

    \Input{Registered peer working window set with enough unfrozen deposit $\widehat{T^w}$, requested sender storage window $T^s=[t_s,t_r]$, transfer time period $|T_t|$.}
    \Output{Selected peer working window list $\widetilde{T^w}$.}
    Initialize $t_{cur}=t_r$, $t_{pre}=t_r$, $T^w_{sel}$\;
    	\While{$t_{pre}>t_s$} {
    		\For{each $T^w \in \widehat{T^w}$} {
            		\If{$T^w.t_b<t_{cur}+|T_t|\ \&$  
                   		$T^w.t_e>t_{cur}+|T_t|\ \&$ \\
                		$\ \ \ T^w.t_b<t_{pre}+|T_t|\ \&$ 
                    	$nonRepeat(T^w)$} {
                		$T^w_{sel} = T^w$;
                    	$t_{pre}=T^w.t_b$\;
                	}
            	}
        
        	\If{$nonRepeat(T^w_{sel})$} {
            	$\widetilde{T^w} \gets T^w_{sel}$; 
                $t_{cur}=t_{pre}$\; 
        	}
            \Else {
            	Fail\;
            }
		}
   
   \end{footnotesize}
    \caption{Peer selection algorithm}
    
    \label{A1} 
\end{algorithm}
The pseudo-code of the peer selection algorithm is shown in Algorithm \ref{A1} (we assume peers have passed registration deadline check and balance check). 
The peer selection problem is decomposed into a series of subproblems. For each subproblem (loop 2-15), the algorithm traverses all available peer working windows $T^w$s (loop 3-8) to find the ones satisfying the conditions that: 1) it covers the input time point of this round (line 4); 2) it has earlier $t_b$ than the ones that have been traversed (line 5); 3) the peer has not been selected for this service (line 5). After an eligible $T^w$ is found, the greedy choice for this round is updated (line 6). Finally, the end of the traversal gives the greedy choice for the current round subproblem. If the greedy choice is different from that of the last round, the algorithm approves it and starts the next round (line 9-11). Otherwise, the algorithm fails to find available $P$s for this service and returns $False$. The complexity of this algorithm is $O(|\widehat{T^w}||\widetilde{T^w}|)$.

\begin{lemma}
\textit{The greedy algorithm that always picks $T^w$ with earliest $t_b$ minimizes the number of selected $P$ for a service.}
\end{lemma}
\begin{proof}
Let us consider that the peer selection problem is decomposed into $n$ rounds of continuous subproblems. If an algorithm falls behind the greedy algorithm in round $i$, then the only way for this algorithm to catch up with the greedy algorithm at round $i+1$ will be to select the greedy choice of round $i+1$ in round $i+1$, but this can at most make its performance same as the greedy algorithm. 
\end{proof}

\end{document}